\documentclass{amsart}
\usepackage{amsfonts}
\usepackage{amsmath}
\usepackage{amssymb}
\usepackage{amsthm}
\usepackage{color}
\usepackage{comment}
\usepackage[foot]{amsaddr}

 \usepackage[utf8]{inputenc} 
\usepackage{pdfpages}

\definecolor{Green}{rgb}{0.20,0.43,0.09}

\newtheorem{Thm}{Theorem}
\newtheorem{Lem}[Thm]{Lemma}
\newtheorem{Prop}[Thm]{Proposition}
\newtheorem{Cor}[Thm]{Corollary}

\theoremstyle{definition}

\theoremstyle{remark}
\newtheorem{Rem}[Thm]{Remark}

\numberwithin{equation}{section}

\def\QQ{{\mathbb{Q}}}

\def\RR{{\mathbb{R}}}
\def\ZZ{{\mathbb{Z}}}
\def\NN{{\mathbb{N}}}

\def\KK{{\mathbb{K}}}
\def\FF{{\mathbb{F}}}

\def\LL{{\mathbb{L}}}

\title[Cake cutting: Explicit examples for impossibility results]{Cake cutting:\\ Explicit examples for impossibility results}
\date{\today}

\author[G.~Ch\`eze]{Guillaume Ch\`eze}
\address{Guillaume Ch\`eze: Institut de Math\'ematiques de Toulouse, UMR 5219\\
Université de Toulouse ; CNRS  \\
UPS IMT, F-31062 Toulouse Cedex 9, France 
}
\email{guillaume.cheze@math.univ-toulouse.fr}

\date{\today}

\begin{document}
	
\begin{abstract}
In this article we suggest a model of computation for the cake cutting problem. 
In this model the mediator can ask the same queries as in the Robertson-Webb model but he or she can only perform algebraic operations as in the Blum-Shub-Smale model. All existing algorithms described in the Robertson-Webb model can be described in this new model.\\
We show that in this model there exist explicit couples of measures for which no algorithm  outputs an equitable fair division with connected parts.\\
We also show that there exist explicit set of  measures for which no algorithm in this model outputs a fair division which maximizes the utilitarian social welfare function.\\
The main tool of our approach is Galois theory.
\end{abstract}

\maketitle

\section*{Introduction}
In 1837, Pierre Wantzel has shown that there exists no general construction using only compass and straightedge which divides an angle into three equal angles. The proof relies on algebra and field theory. The angle trisection problem can be seen as a fair division problem: we have a portion of pizza and we want to divide it  in a fair way between three friends (by using only compass and straightedge constructions\ldots). Wantzel's theorem says that this problem has no solution.\\

In this article, we are going to study a similar fair division problem and we are going to use similar tools.\\

In the following, we consider an heterogeneous good, for example: a cake, land, time or computer memory, represented by the interval $X=[0,1]$ and $n$ players with different points of view. We associate to each player a non-atomic probability measure $\mu_i$ on the interval $X=[0;1]$. These measures represent the utility functions of the player. This means that if $[a,b] \subset X$ is a part of the cake then $\mu_i([a,b])$ is the value associated by the $i$-th player to this part of the cake. As $\mu_i$ are probability measures, we have $\mu_i(X)=1$ for all  $i$.\\
A division of $X$ is a partition $X=\sqcup_i X_i$ where $X_i$ is the part given to the $i$-th player. A division is \emph{simple} when each $X_i$ is an interval.\\

 Several notions of fair division exists.\\
We say that a division is \emph{proportional} when $\mu_i(X_i) \geq 1/n$.\\
We say that a division is \emph{envy-free} when for $i\neq j$, we have $\mu_i(X_i) \geq \mu_i(X_j)$.\\
We say that a division is \emph{equitable} when for all $i \neq j$, we have $\mu_i(X_i)=\mu_j(X_j)$.\\

We say  that the  division  $X=\sqcup_i X_i$ \emph{maximizes the utilitarian social welfare function} when
$$\sum_{i=1}^n \mu_i(X_i) \geq \sum_{i=1}^n \mu_i(Y_i),$$
 for all partition $X=\sqcup_i Y_i$.\\

 The problem of fair division (theoretical existence of fair division and construction of algorithms) has been studied in several papers \cite{Steinhaus,DubinsSpanier, EvenPaz, EdmondsPruhs, BramsTaylorarticle, RoberstonWebbarticle, Pikhurko, Thomson2006, Procacciasurvey, BJK, AzizMackenzie}, and books about this topic, see e.g. \cite{RobertsonWebb,BramsTaylor, Procacciachapter,Barbanel}. These results appear in the mathematics, economics, political science, artificial intelligence and computer science literature. Recently, the cake cutting problem has been studied intensively by computer scientists for solving resource allocation problems in multi agents systems, see e.g.~\cite{Chevaleyre06,Chen,Dynamic,Branzei}. \\

 A practical problem is the computation of fair divisions. In order to describe algorithms we thus need a model of computation. There exist two main classes of cake cutting algorithms: discrete and continuous protocols (also called moving knife methods). Here, we study only discrete algorithms. These kinds of algorithms can be  described thanks to the  classical model introduced by Robertson and Webb and formalized by Woeginger and Sgall in \cite{Woeg}. In this model we suppose that a mediator interacts with the agents. The mediator asks two type of queries: either cutting a piece with a given value, or evaluating a given piece. More precisely, the two type of queries allowed are:
\begin{enumerate}
\item $eval_i(x,y)$: Ask agent $i$ to evaluate the interval $[x,y]$. This means compute $\mu_i([x,y])$.
\item $cut_i(x,a)$: Asks agent $i$ to cut a piece of cake $[x,y]$ such that $\mu_i([x,y])=a$. This means: for given $x$ and $a$, solve $\mu_i([x,y])=a$.
\end{enumerate} 
In the Robertson-Webb model the mediator can adapt the queries from the previous answers given by the players. In this model, the complexity counts the finite number of queries necessary to get a fair division. For a rigorous description of this model we can consult: \cite{Woeg,Branzei2017}.\\

The result of a query is a real number and thus the mediator has to manipulate real numbers. There exist two possible models of computation which allows this task.\\
First, we can consider computable real numbers. Roughly speaking a real number is said to be computable if there exists a Turing machine which writes digit by digit the (infinite) decimal expansion of this number.  Unfortunately, this model of computation is not natural in our setting because we cannot decide in this model if a computable number is equal to 0. This means that we cannot decide if two computable numbers are equal. Thus, with this model, the mediator cannot check if a fair division is equitable.\\
Second, we can consider the BSS model of computation. This model has been  developed by Blum, Shub and Smale (BSS). It allows to study algorithms over a ring. Roughly speaking a BSS machine has registers which can hold arbitrary elements of the studied ring (here $\RR$), and perform exact arithmetic ($+,-,\times,\div$) and can branch on conditions based on exact comparisons ($=, <,>, \leq, \geq$). Furthermore, with this theory when the ring is $\ZZ/2\ZZ$ then we recover the classical theory of Turing machine. For a detailed description of this model see \cite{BSS,BCSS}. \\

In this article we are going to suppose that the mediator use a BSS machine. We call this new model of computation the BSSRW model (Blum-Shub-Smale-Robertson-Webb model)  and we are going to prove impossibilty results. \\

In the fair division literature some impossibility results have been already given.\\
 Stromquist in \cite{Stromquist} has proved that there exists no algorithm giving a simple and envy-free fair division for $n\geq 3$ players. When $n=2$, the classical ``Cut and Choose" algorithm gives a simple and envy-free fair division.\\
 Cechl\'arov\'a et al.  have shown, in \cite{Cech}, that there exists no algorithm computing a simple and equitable fair division for $n \geq 3$ players in the Roberston-Webb model.\\
 
  The strategy used in these articles is the following: they suppose that an algorithm computing the desired division exists and then by an iteration process they construct from this algorithm a set of measures giving a contradiction. Thus they obtain a result of this kind: for all algorithms in the Roberston-Webb model there exists a set of measures for which the desired fair division cannot be given.\\
It must be noticed that this approach gives for each algorithm a set of measures leading to a contradiction. Thus the set of measures is related to the algorithm. Moreover, the measures are not explicitly given. Therefore, we can imagine that these sets of measure correspond to  very complicated situations not appearing in practice and that for ``reasonable" sets of measures the contradiction does not occur.\\

 Procaccia and Wang have also given an impossibility result for equitable fair division in \cite{ProcWang}. As a corollary of a theorem about a lower bound for equitable division they deduce that there exists no  algorithm giving an equitable fair division. However, with this approach we still cannot give an explicit example of measures such that no algorithm in the Robertson-Webb model returns an equitable division with this input.\\
 
In the first part of this article,  we are going to study simple equitable fair divisions. This topic has been less studied than proportional and envy-free divisions. However, there exist some results showing the existence of such fair divisions \cite{Cechexistence,Segal-Halevi,Chezeequitable}. Furthermore, if we consider a continous protocol, it is possible to get an equitable fair division (not necessarily simple) thanks to Austin's moving knife procedure, see \cite{Austin}.\\
Here, we  are going to give explicit examples where two players cannot get an equitable fair division with connected parts if we use our suggested model of computation.\\


\begin{Thm}\label{thm:1}
In the BSSRW model of computation no algorithm returns a simple and equitable division when the measures $(\mu_1, \mu_2)$ are given by 
$$\mu_1\big([0,x]\big)=x,\quad \mu_2\big([0,x]\big)=x^5.$$
\end{Thm}

The strategy used to prove this theorem is the following: We are going to show that if there exists an equitable and simple division $X=[0,t] \sqcup [t,1]$ then the final cutpoint $t$ must satisfy a polynomial equation. Then, with elementary field theory, we can show that $t$ cannot be computed with the BSSRW model.\\

Now, if we use Abel's impossibility theorem and Galois' theory showing that some polynomials are not solved by radicals, then we obtain other examples as stated in the next theorem:

\begin{Thm}\label{thm:2}
In the BSSRW model of computation there exist measures $(\mu_1,\mu_2)$ such that no algorithm returns a simple and equitable division for these measures.\\
Furthermore, we can take  $(\mu_1, \mu_2)$ in the following way: 
$$\mu_1\big([0,x]\big)=x,\quad \mu_2\big([0,x]\big)=x^d$$
where 
\begin{itemize}
\item  $d\geq 5$ is even,
\item or $d \geq 5$ is odd with $d \not \equiv 2 \, [3]$,
\item or  $d \geq 5$ is prime and $d  \equiv 2 \, [3]$.
\end{itemize}
\end{Thm}

Thus, when we have two players, we can give \emph{easy and explicit} couples of measures for which no algorithm in the BSSRW model gives a simple and equitable fair division. \\

In a second part, we show that in the BSSRW model we cannot obtain  a fair division which maximizes the utilitarian social welfare function. In this last situation, we will consider $n$ players and we will not suppose the division $X=\sqcup_i X_i$ to be simple. \\
\begin{Thm}
In the BSSRW model of computation there exists measures $\mu_1, \mu_2, \ldots, \mu_n$ such that no algorithm returns a division which maximizes the utilitarian social welfare function.\\
Furthermore, we can take $\mu_1, \mu_2, \ldots, \mu_n$ in the following way:
$$\mu_1([0,x])=\cdots=\mu_{n-1}([0,x])=x,\quad  \mu_n([0,x])=x^{p}$$
where $p\geq 3$ is a prime number.
\end{Thm}

Now, in order to state our results, we introduce our model of computation.

\section{The BSSRW model}
In the Robertson-Webb model of computation the computational power of the mediator is not specified. It is not mentioned what kind of computations the mediator can perform with the results of the queries. Furthermore, the number of elementary operations done by the mediator (equality and inequality tests and arithmetic operations $+,-,\times,\div$) is not taken into account in the complexity. This point has been discussed in \cite{Chezecry}.\\
 Here, we suppose as in the classical model that the mediator can use the $cut_i$ and $eval_i$ queries. However, we also suppose that the mediator can  \emph{only} perform equality and inequality tests and the usual algebraic operations: $+,-, \times, \div$ on the results of queries. We also suppose that the mediator can use freely the rational numbers. This means that the mediator uses a Blum-Shub-Smale machine.\\

These assumptions are not restrictive. Indeed, no known algorithm uses the computation of a logarithm or of an exponential by the mediator or more generally the computation of a transcendental function.\\
 Furthermore, when the mediator needs a constant during the algorithm this constant is always a rational number. Indeed, in practice the mediator never asks a query of the form $cut_i(0,\frac{e}{4}\mu_i(X))$, where $e=2,718\ldots$ is Napier's constant.  Queries  have the form $cut_i\big(0,\frac{\mu_i(X)}{n}\big)$ or $cut_i\big(0,\frac{p}{q}\mu_i(X)\big)$, where $p,q,n$ are integers.\\
 
Therefore, if we suppose that the answer to the first three queries are  denoted by $\alpha_1, \alpha_2, \alpha_3$,  then in this new model, the fourth query is of the form $cut_i(\beta_1,\beta_2)$ or $eval_i(\beta_1,\beta_1)$ where  $\beta_1, \beta_2 \in \QQ(\alpha_1,\alpha_2,\alpha_3)$. This means that $\beta_1$ and $\beta_2$ are rational  expressions in terms of $\alpha_1,\alpha_2,\alpha_3$.\\

The algebraic assumption is not restrictive and to author's knowledge all algorithms written in the classical Robertson-Webb model can be written in this Blum-Shub-Smale-Robertson-Webb model (BSSRW model). However, these precisions are important for our study. Indeed, if the algorithm uses $k$ queries with answers $\alpha_1, \ldots, \alpha_k$ for computing a fair division, then the cutpoints used in the output of the algorithm must belong to $\QQ(\alpha_1,\alpha_2,\ldots,\alpha_k)$. This gives an algebraic condition for the final cutpoints. Using this algebraic condition, we can prove our impossibility theorems.

\subsection*{Notations and elementary results}
For given measures $\mu_1,\mu_2,\ldots, \mu_n$ we denote by $f_i$, $i=1,\ldots, n$ the function $$f_i(x)=\mu_i\big([0,x]).$$
Let $\alpha_j$ be the result of the $j$-th query, then we set 
$$\KK_j=\QQ(\alpha_1,\ldots,\alpha_j).$$
We thus have $\KK_j=\KK_{j-1}(\alpha_j)$ and $\KK_{0}=\QQ$.\\

We recall that when a field $\FF$ is a subfield of a field $\KK$ then we say that we have a field extension and this is denoted by $\KK/ \FF$. Furthermore, the dimension of $\KK$ seen as a $\FF$-vector space is called the degree of the extension and is denoted by $[\KK:\FF]$. When the degree is finite we say that the extension is finite. Moreover, when we have the inclusion $\FF \subset \KK \subset \LL$, this gives two extensions $\LL/ \KK$ and $\KK/ \FF$. If the degree of these two extensions are finite then the extension $\LL/ \FF$ is also finite and we have the following equality: $[\LL:\FF]=[\LL:\KK][\KK:\FF]$, see e.g. \cite[Lemma~15.3]{Tignol}.\\
Furthermore, we recall that if $\alpha$ is a root of an irreducible polynomial in $\KK[T]$ with degree $d$ then $[\KK(\alpha):\KK]=d$, see \cite[Proposition 12.15]{Tignol}.\\

Now, we recall some classical results which will be useful in our proofs.
\begin{Lem}\label{lem:0}
Let $p$ be a prime number and let $b$ be an element of some field $\FF$, which is not a $p$-th power in $\FF$.\\ 
The polynomial $T^p-b$ is irreducible over $\FF$.
\end{Lem}
\begin{proof}
For a proof see \cite[Lemma 13.9]{Tignol}.
\end{proof}

\begin{Lem}\label{lem:1}
If $f_1(x)=\cdots=f_{n-1}(x)=x$, and  $f_n(x)=x^{p}$ with $p \geq 3$  a prime number then the degree of the field extension $\KK_j/\KK_{j-1}$ is equal to $p$ or $1$.
\end{Lem}
\begin{proof}
By definition we have $\KK_j=\KK_{j-1}(\alpha_j)$.\\
Two situations appear:\\
First, $\alpha_j=eval_i(x,y)$, where $x,y \in \KK_{j-1}$.\\
 As,  $eval_1(x,y)=\cdots= eval_{n-1}(x,y)=y-x$ and $eval_2(x,y)=y^p-x^p$, we deduce that in this case $\alpha_j \in \KK_{j-1}$. Thus $\KK_j=\KK_{j-1}$ and the degree of the extension is equal to one.\\
Second, $\alpha_j=cut_i(x,a)$, where $x,a \in \KK_{j-1}$. \\
If $i=1, \ldots,n-1$ then $\alpha_j$ is the solution of $\alpha_j - x=a$ then $\alpha_j \in \KK_{j-1}$ and the degree of the field extension is equal to one.\\
If $i=n$ then $\alpha_j$ is the solution of $\alpha_j^p - x^p=a$.\\
 If this equation has a solution in $\KK_{j-1}$, this  means $\alpha_j \in \KK_{j-1}$ and the degree of $\KK_j/\KK{j-1}$ is equal to one.\\
If the equation $\alpha_j^p - x^p=a$ has no solution in $\KK_{j-1}$, then $a+x^p$ is not a $p$-th power in $\KK_{j-1}$. Thus by Lemma~\ref{lem:0} the polynomial $T^p-x^p-a \in \KK_{j-1}[T]$ is irreducible over $\KK_{j-1}$. In this last case, the degree of the extension $\KK_j/\KK_{j-1}$ is equal to $p$, see \cite[Proposition~12.15]{Tignol}.
\end{proof}

\section{Impossible equitable fair divisions}

\subsection{Proof of Theorem~\ref{thm:1}}
The idea used to prove our theorems is the following:\\
If there exists a simple and equitable fair division $X=[0,t]\sqcup [t,1]$, then we have 
$$\mu_1\big([0,t]\big)=\mu_2\big([t,1]\big) \textrm{ or } \mu_2\big([0,t]\big)=\mu_1\big([t,1]\big).$$
This gives
$$f_1(t)=1-f_2(t) \textrm{ or } f_2(t)=1-f_1(t).$$
These two equations are equivalent to the following one:
$$(E) \quad f_1(t)+f_2(t)-1=0.$$
This equation gives an algebraic condition on $t$ which cannot be satisfied when we use the BSSRW model. The following proof explains why. \\

\begin{proof}[Proof of Theorem~\ref{thm:1}]
If an algorithm in the algebraic Robertson-Webb model computes an equitable and simple fair division in $k$ steps then the final cutpoint $t$ belongs to $\KK_k=\QQ(\alpha_1,\ldots,\alpha_k)$. We have thus the inclusion $\QQ \subset \QQ(t) \subset \KK_k$.\\
As $t$ satisfies the equation $(E)$ we have here
$$t^5+t-1=0.$$
We can factorize this expression and we obtain:
$$t^5+t-1=(t^2-t+1)(t^3+t^2-1)=0.$$
As the polynomial $T^2-T+1$ has no real roots we deduce that we have
$$t^3+t^2-1=0.$$
Furthermore, the polynomial $T^3+T^2-1$ is irreducible in $\QQ[T]$. Indeed, it suffices to remark that $T^3+T^2-1$ has no rational roots. This a consequence of the rational roots theorem. In our situation, this theorem says that if $T^3+T^2-1$ has a rational root then  it must be $\pm 1$. Thus $T^3+T-1$ has no rational root and this polynomial is irreducible over $\QQ$.\\

 We deduce then
$$[\QQ(t):\QQ]=3.$$
However, by Lemma~\ref{lem:1} with $n=2$, we have 
$$[\KK_k:\QQ]=5^l,$$ 
with $l \leq k$.
Therefore, the equality
$$5^l=[\KK_k:\QQ]=[\KK_k:\QQ(t)][\QQ(t):\QQ]=[\KK_k:\QQ(t)]\times 3$$
is impossible and this concludes the proof.
\end{proof}

\begin{Rem}
We can remark that with the measures given in Theorem~\ref{thm:1} even if the mediator can compute square roots then we still have an impossibility result.\\
Indeed, if the mediator use a square root after the $j$-th query then this means that during the algorithm the mediator uses a number $\alpha= \sqrt{\beta}$, where $\beta \in \KK_{j}$. Then, we must consider some extensions $\KK'_j=\KK_{j}(\alpha)$ where $\alpha^2 =\beta \in \KK_{j}$. Therefore, the degree of some extensions $\KK_j/\KK_{j-1}$ can be equal to two. Thus, in the previous proof the degree $[\KK_k:\QQ]$ has the following form $[\KK_k:\QQ]=2^m5^l$ and the conclusion is still valid.
\end{Rem}



\subsection{Proof of  Theorem~\ref{thm:2}}
In order to prove Theorem~\ref{thm:2}, we need some tools.

\begin{Lem}\label{lem:2}
If $f_i(x)=x^{e_i}$, for $i=1,2$ then for all $j\geq 1$, $\KK_j=\KK_{j-1}$ or $\KK_j$ is a radical extension of $\KK_{j-1}$.\\
\end{Lem}

Roughly speaking, this lemma says that the field $\KK_j$ is of the following form: $\KK_{j-1}(\sqrt[n]{\alpha})$, where $n$ is an integer and $\alpha\in\KK_{j-1}$.

\begin{proof}
If the $j$-th query is of the form $eval_i\big([x,y]\big)$ with $x,y \in \KK_{j-1}$ then $eval_i\big([x,y]\big)$ is equal to $y^{e_i}-x^{e_i}$. Thus the result to this query $\alpha_j=y^{e_i}-x^{e_i} \in \KK_{j-1}$. In this situation we have then $\KK_{j}:=\KK_{j-1}$.\\
If the $j$-th query is of the form $cut_i\big(x,a\big)$ with   $x,a \in \KK_{j-1}$ then the result $\alpha_j$ to this query is the unique solution in $[0,1]$ of the following equation:
$$\alpha_j^{e_i}-x^{e_i}=a.$$ 
This implies $\alpha_j=\sqrt[e_i]{a+x^{e_i}}$ and $\KK_j:=\KK_{j-1}(\sqrt[e_i]{a+x^{e_i}})$.\\
The extension $\KK_{j}/ \KK_{j-1}$ is thus a radical extension.\\
\end{proof}

As $\QQ \subset \KK_1 \subset \cdots \subset \KK_k$ we have by definition of a radical extension, see \cite[Chapter 13]{Tignol}, the following corollary:
\begin{Cor}\label{cor}
For all $j \geq 1$, the extension $\KK_j/\QQ$ is radical.
\end{Cor}

Now, we recall a result about the irreducibility and the Galois group of certain trinomials.

\begin{Prop}[Selmer \cite{Selmer}]\label{prop:selmer}
The polynomials $T^d-T-1$ are irreducible in $\QQ[T]$ for all $d$.\\
The polynomials $T^d+T+1$ are irreducible in $\QQ[T]$ for $d\not \equiv\, 2\, [3]$, but have a factor $T^2+T+1$ when $d \equiv 2 \, [3]$. In the latter case, $T^d+T+1$ has another factor which is irreducible.
\end{Prop}

\begin{Prop}[Osada \cite{Osada}]\label{prop:osada}
Let $f(T)=T^d+aT+b \in \ZZ[T]$, where $a=a_0c^d$ and $b=b_0c^d$ for some integer $c$. Then the Galois group over $\QQ$ of this polynomial is isomorphic to the symmetric group $\mathcal{S}_d$ if the following conditions are satisfied:
\begin{enumerate}
\item $f(T)$ is irreducible over $\QQ$,
\item $\gcd\big(a_0c(d-1),db_0\big)=1$.
\end{enumerate}
\end{Prop}

These propositions allow us to show the following lemma.

\begin{Lem}\label{lem:groupegalois}
If $d$ is even or if $d\geq 5$ is odd with $d\not\equiv 2[3]$, then the Galois group over $\QQ$ of $T^d+T-1$ is isomorphic to the symmetric group $\mathcal{S}_d$.\\
If $d\geq 5$ is prime and $d\equiv 2[3]$ then $T^d+T-1$ is reducible over $\QQ$: it has an irreducible factor with degree $2$ and another one with degree $d-2$.
\end{Lem}

\begin{proof}
When $d$ is even, the change of variables $Y=-T$ gives 
$$T^d+T-1=(-T)^d-(-T)-1=Y^d-Y-1.$$ 
We deduce that $T^d+T-1$ is irreducible since, by Proposition \ref{prop:selmer}, $Y^d-Y-1$ is irreducible.\\
When $d\geq 5$ is odd with $d\not \equiv\, 2\, [3]$, the change of variables $Y=-T$ gives
$$T^d+T-1=-(-T)^d-(-T)-1=-Y^d-Y-1=-(Y^d+Y+1).$$
As before, thanks to Proposition~\ref{prop:selmer}, we deduce that $T^d+T-1$ is irreducible since $Y^d+Y+1$ is irreducible.\\

Therefore, in the two previous cases  $T^d+T-1$ is irreducible over $\QQ$. \\ Proposition~\ref{prop:osada} with $a_0=c=1$ and $b=-1$ entails in these situations  that the Galois group of  $T^d+T-1$ is isomorphic to $\mathcal{S}_d$.\\

When $d\geq 5$ is prime and $d\equiv 2[3]$, the change of variables $Y=-T$ gives as before $T^d+T-1=-(Y^d+Y+1)$
 and Proposition \ref{prop:selmer} gives the desired result.
\end{proof}

Now, we can prove Theorem~\ref{thm:2}.
\begin{proof}[Proof of Theorem~\ref{thm:2}]
 We  suppose that there exists an algorithm in the BSSRW model computing an equitable and simple fair division $X=[0,t] \cup [t,1]$. Then, $t$ must satisfy the equation $(E)$. Here, this equation is:
$$t^d+t-1=0.$$

 First,  we suppose that $d$ satisfies the hypothesis of one of the first two items.\\
 
As $t$ must belong to $\KK_k$ and,  by Corollary~\ref{cor}, $\KK_k$ is a radical extension of $\QQ$, we deduce that $t$ has a radical expression over $\QQ$. Thus as the polynomial $T^d+T-1$ is irreducible then it can be solved by radicals over $\QQ$, see  \cite[Proposition 14.33]{Tignol}.\\
However,  by Lemma~\ref{lem:groupegalois}, the Galois group of $T^d+T-1$ is isomorphic to $\mathcal{S}_d$ . 
Then Galois' theory implies that this polynomial cannot be solved by radicals over $\QQ$, see \cite[Chapter 14]{Tignol}. This gives the desired contradiction.\\

Now, we suppose that $d \geq 5$ is prime and $d \equiv 2 \, [3]$.\\
 In this case, the proof is a generalization of the proof of Theorem \ref{thm:1}.\\

By Lemma \ref{lem:groupegalois}, the polynomial $T^d+T-1$ has an irreducible factor with degree $2$ and another one with degree $d-2$. This gives
$$[\QQ(t):\QQ]=2 \textrm{ or } [\QQ(t):\QQ]=d-2.$$
Furthermore, thanks to Lemma \ref{lem:1} we have
$$[\KK_k:\QQ]=d^l,$$
where $l \in \NN$.\\
The equality $$[\KK_k:\QQ]=[\KK_k:\QQ(t)][\QQ(t):\QQ]$$
is then impossible since $d$ is prime. This concludes the proof.
\end{proof}

\section{Impossibility to maximize the social welfare function}
\setcounter{Thm}{2}
\begin{Thm}
In the BSSRW model of computation there exists measures $\mu_1, \mu_2, \ldots, \mu_n$ such that no algorithm returns a division which maximizes the utilitarian social welfare function.\\
Furthermore, we can take $\mu_1, \mu_2, \ldots, \mu_n$ in the following way:
$$\mu_1([0,x])=\cdots=\mu_{n-1}([0,x])=x,\quad  \mu_n([0,x])=x^{p}$$
where $p\geq 3$ is a prime number.
\end{Thm}
It must be noticed that the division is not supposed to be simple.\\

The previous theorem deals with a problem involving an inequality about the utilitarian social welfare function. Our strategy will be to reduce this problem to a problem involving an equation.\\
In general, fair division problems are stated with inequalities, see e.g. envy-free division and  proportional division. We can always reduce these problems to problems involving equalities. For example, the condition $\mu_i(X_i) \geq 1/n$ becomes $\mu_i(X_i)=1/n+e^2$, where $e \in \RR$. However, with this method we introduce new variables and the problem do not become necessarily easier with these equalities.

\begin{proof}
Let  $X=\sqcup_{i=1}^n X_i$ be a division of $X$ constructed with $m$ cuts.\\
 This means that each $X_i$ can be written in the following way $X_i=\sqcup_{j \in I_i}[x_{j},x_{j+1}]$, and we have $m$ different $x_j$: $x_1 \leq x_2 \leq \cdots \leq x_m$.\\

The value of the utilitarian social welfare function associated to this division is 
$$\mathcal{F}(x_1,\ldots,x_m)=\sum_{i=1}^{n-1} \sum_{j \in I_i}(x_{j+1}-x_j)+\sum_{j \in I_n} (x_{j+1}^{p}-x_{j}^{p}).$$

Now, we consider an index $j_0\in I_n$.\\
We remark that 
$$\mathcal{F}(x_1,\ldots,x_m)=g(x_1,\ldots,x_{j_0-1},x_{j_0+1},\ldots,x_m)+x_{j_0}-x_{j_0}^p,$$
where $g$ is a function independent of $x_{j_0}$.\\
If this division maximizes the social welfare function then we must have 
$$\partial_{x_{j_0}}\mathcal{F}(x_1,\ldots,x_m)=1-px_{j_0}^{p-1}=0.$$
Then $x_{j_0}=\sqrt[p-1]{1/p}$. As $p\geq 3$ is a prime number we deduce that $x_{j_0} \not \in \QQ$. Furthermore, as $x_{j_0}$  is a root of the polynomial $pT^{p-1}-1$ we get
$$1<[\QQ(x_{j_0}):\QQ]< p.$$
Now, suppose that an algorithm in the BSSRW model computes in $k$ steps a division which maximizes the utilitarian social welfare function. Then, $x_{j_0} \in \KK_k$.
However, by Lemma \ref{lem:1}, we have
$$[\KK_k:\QQ]=p^l,$$
with $l \leq k$. Therefore, the equality
$$[\KK_k:\QQ]=[\KK_k:\QQ(x_{j_0})][\QQ(x_{j_0}):\QQ]$$
is impossible since $p$ is prime. This concludes the proof.
\end{proof}


\textbf{Acknowledgement}
The author thanks Erel Segal-Halevi for his precious remarks about a preliminary version of this article.



 

\newcommand{\etalchar}[1]{$^{#1}$}

\end{document}